\documentclass[12pt]{article}
\usepackage{amsmath,amssymb,amsthm}
\usepackage{setspace}
\usepackage[margin=1in]{geometry}
\usepackage{microtype}
\usepackage{mathptmx}
\usepackage[hidelinks]{hyperref}
\newtheorem{theorem}{Theorem}
\newtheorem{proposition}{Proposition}
\newtheorem{corollary}{Corollary}
\newtheorem{remark}{Remark}

\newcommand{\R}{\mathbb{R}}
\newcommand{\N}{\mathbb{N}}

\title{From Counting to Continuous Natural Exponential Families}
\author{Shaul K. Bar-Lev\\
Faculty of Industrial Engineering and Technology Management\\
HIT -- Holon Institute of Technology, Holon, Israel}
\date{}

\begin{document}
\maketitle

\begin{abstract}
Natural exponential families (NEFs) are characterized by their variance functions (VFs), that is, by the pair $(V,M)$, where $V$ is the variance expressed as a function of the mean and $M$ is the mean domain. This raises the inverse question of when a given $V$ is the VF of an NEF and what properties of the family follow from it. Bar-Lev (1987), as part of a more general result for absolutely monotone functions, showed that every nonzero polynomial $V(m)=\sum_{j=1}^{r}a_jm^j$, $a_j\geq0$, is the VF of an infinitely divisible NEF on the positive mean domain. This polynomial class splits exhaustively according as $a_1>0$ or $a_1=0$. Bar-Lev, Letac and Ridder (2024) showed that, after scaling so that $a_1=1$, the first case is a counting NEF on $\mathbb N_0$. We prove the complementary result: if $a_1=0$, the NEF is absolutely continuous with respect to Lebesgue measure on the positive half-line. More generally, we prove absolute continuity for $V(m)=m^2G(m)$ when $G$ is a nonzero power series with nonnegative coefficients. For polynomial $G$, the mean domain is the whole positive half-line. Brief cumulant and inferential consequences are also noted.
\end{abstract}



\section{Introduction}

Natural exponential families (NEFs) are characterized, up to the usual equivalence of generating measures, by their variance functions together with their mean domains. If $V(m)$ is the variance expressed as a function of the mean $m$ and $M$ is the mean domain, then the pair $(V,M)$ determines the NEF; see Barndorff-Nielsen \cite{BarndorffNielsen1978} and Morris \cite{Morris1982}. This raises the natural inverse question: when is a given positive function $V$ the variance function of an NEF, and what properties of the corresponding family can be deduced from $V$?

We recall only the notation needed below. Let $\mu$ be a positive Radon measure on $\R$, not concentrated at one point, and put
\[
L_\mu(\theta)=\int_{\R}e^{\theta x}\,\mu(dx),\qquad
\Theta_\mu=\operatorname{int}\{\theta:L_\mu(\theta)<\infty\},\qquad
k_\mu(\theta)=\log L_\mu(\theta).
\]
Assuming $\Theta_\mu\neq\varnothing$, the NEF generated by $\mu$ is
\[
F(\mu)=\{P_\theta(dx)=\exp\{\theta x-k_\mu(\theta)\}\mu(dx):\theta\in\Theta_\mu\}.
\]
Write $S_\mu=\operatorname{supp}(\mu)$ and $C_\mu=\overline{\operatorname{conv}}(S_\mu)$. Exponential tilting leaves $S_\mu$, and hence $C_\mu$, unchanged. The mean map, mean domain and variance function satisfy
\begin{equation}
 m(\theta)=k_\mu'(\theta),\qquad M_F=k_\mu'(\Theta_\mu),\qquad
 V(m)=k_\mu''(\theta(m)),\qquad \frac{d\theta(m)}{dm}=\frac{1}{V(m)},
\label{eq:meanVF}
\end{equation}
where $\theta(m)$ is the inverse mean map. The family is steep when $M_F=\operatorname{int}C_\mu$. The distinction between mean domain and convex support will be used below.

Bar-Lev \cite{BarLev1987}, in his discussion of J\o rgensen's paper \cite{Jorgensen1987}, gave a general sufficient condition for a positive absolutely monotone function to be the variance function of an infinitely divisible NEF. Thus the 1987 result is substantially more general than the polynomial case. An important consequence is that every nonzero polynomial
\begin{equation}
V(m)=\sum_{j=1}^{r}a_jm^j,\qquad a_j\geq0,
\label{eq:BLpoly}
\end{equation}
with zero constant term is an infinitely divisible variance function on the positive mean domain. This large class contains classical families and polynomial variance functions of arbitrary degree; in particular, it intersects the quadratic class of Morris \cite{Morris1982} and the cubic class of Letac and Mora \cite{LetacMora1990}.

Within the general 1987 framework, two natural subclasses are
\[
V_1(m)=mG(m),\qquad V_2(m)=m^2G(m),
\]
with $G$ nonzero and absolutely monotone, subject to the boundary condition recalled in Section 2. These two forms are not asserted to exhaust the full 1987 construction. They do, however, give an exhaustive dichotomy for the polynomial class \eqref{eq:BLpoly}: if $a_1>0$, then
\[
V(m)=m\left(a_1+a_2m+\cdots+a_rm^{r-1}\right),
\]
whereas if $a_1=0$, then
\[
V(m)=m^2\left(a_2+a_3m+\cdots+a_rm^{r-2}\right).
\]
Thus the polynomial alternative $a_1>0$ versus $a_1=0$ is the finite-dimensional counterpart of the two subclasses $mG(m)$ and $m^2G(m)$.

Bar-Lev, Letac and Ridder \cite{BarLevLetacRidder2024} resolved the first subclass. The required normalization is simply a scale change. If $Y=cX$, then
\[
V_Y(m)=c^2V_X(m/c)=\sum_{j=1}^r a_jc^{2-j}m^j.
\]
Hence, when $a_1>0$, taking $c=1/a_1$ makes the linear coefficient equal to one. Their result then shows that the normalized family is a counting NEF on $\N_0$, with positive masses at zero and one; consequently the original family is supported on the lattice $a_1\N_0$. Notice that the quadratic coefficient is invariant under scaling, so $a_2$ cannot in general be normalized in the same way. The counting subclass contains the Poisson and negative-binomial families and, through the cubic classification of Letac and Mora \cite{LetacMora1990}, the Abel, Tak\'acs, strict-arcsine and large-arcsine types.

The present paper establishes the complementary result for the second subclass. We take
\begin{equation}
G(m)=\sum_{j\geq0}a_jm^j,\qquad a_j\geq0,
\label{eq:Gseries}
\end{equation}
with positive radius of convergence $R$, and $V(m)=m^2G(m)$. Existence and infinite divisibility follow from the 1987 criterion; the new issue is support type. We prove that every member of the corresponding NEF is carried by $[0,\infty)$ and is absolutely continuous with respect to Lebesgue measure. Consequently, for the polynomial class \eqref{eq:BLpoly} there is a sharp dichotomy: $a_1>0$ gives, after scaling, the counting case, whereas $a_1=0$ gives the absolutely continuous case. For polynomial $G$, $R=\infty$ and $M_F=(0,\infty)$; for finite-radius nonpolynomial $G$, $M_F=(0,R)$ while $C_\mu=[0,\infty)$.

The construction therefore yields an infinite class of absolutely continuous NEFs on the positive half-line, providing a broad collection of models for statistical inference and conjugate Bayesian analysis directly through their variance functions, even when an explicit density is unavailable. The class contains important classical families as special cases. Among the six NEFs with quadratic variance functions classified by Morris \cite{Morris1982}, it contains the gamma family, the positive-half-line absolutely continuous member relevant here. In the cubic classification of Letac and Mora \cite{LetacMora1990}, it contains the inverse Gaussian family, with variance function proportional to $m^3$, and the Kendall--Ressel family, with variance function of the form $m^2(1+cm)$; the Abel, Tak\'acs, strict-arcsine and large-arcsine cubic families belong instead to the discrete counting branch. Thus the present construction not only produces infinitely many new continuous models but also places several classical quadratic and cubic examples in one variance-function framework.

The paper is organized as follows. Section 2 recalls the 1987 existence result and the discrete branch. Section 3 proves the continuous-branch theorem. Section 4 gives polynomial consequences, examples, cumulants, and brief statistical and conjugate-Bayesian remarks. Section 5 concludes.

\section{The earlier existence result and the discrete branch}

The following result is the existence input. It is the form of Bar-Lev's 1987 criterion needed in this paper.

\begin{proposition}
Let $0<R\leq\infty$ and let $V$ be positive and absolutely monotone on $(0,R)$. If, for some $m_0\in(0,R)$,
\begin{equation}
\int_0^{m_0}\frac{du}{V(u)}=\infty,
\label{eq:BLcondition}
\end{equation}
then $(V,(0,R))$ is the variance function of an infinitely divisible NEF.
\end{proposition}

For $G$ in \eqref{eq:Gseries}, $m^2G(m)$ is absolutely monotone. If $r=\min\{j:a_j>0\}$, then
\[
V(m)=m^2G(m)\sim a_rm^{r+2}\qquad(m\downarrow0),
\]
so \eqref{eq:BLcondition} holds and the corresponding infinitely divisible NEF exists.

Bar-Lev, Letac and Ridder \cite{BarLevLetacRidder2024} showed that the normalized first-order form $mG(m)$ gives an infinitely divisible counting NEF on $\N_0$. We now establish the complementary continuous side.

\section{The continuous branch}

We now prove the principal result.

\begin{theorem}
Let
\[
G(m)=\sum_{j\geq0}a_jm^j,
\qquad a_j\geq0,
\]
have radius of convergence $R>0$ and be nonzero. Let $F$ be the infinitely divisible NEF whose variance function is
\[
V(m)=m^2G(m),
\qquad 0<m<R,
\]
whose existence follows from the preceding proposition. Then every member of $F$ is carried by $[0,\infty)$ and is absolutely continuous with respect to Lebesgue measure. In particular, no member has an atom at zero.
\end{theorem}

\begin{proof}
Let $r=\min\{j:a_j>0\}$ and write
\[
G(m)=m^rH(m),
\]
where $H$ is analytic near zero and $H(0)=a_r>0$. Then
\begin{equation}
V(m)=m^pH(m),
\qquad p=r+2\geq2.
\label{eq:pform}
\end{equation}
Choose $m_0\in(0,R)$ sufficiently small and let $P$ be the member of $F$ having mean $m_0$. By \eqref{eq:meanVF},
\[
\theta(m)-\theta(m_0)=\int_{m_0}^{m}\frac{du}{V(u)}.
\]
Since $V(u)\sim a_ru^p$ as $u\downarrow0$, the right side tends to $-\infty$. Hence all negative exponential tilts of $P$ exist.

For $s\geq0$, put
\[
L(s)=E_P(e^{-sX}),
\qquad
\Phi(s)=-\log L(s),
\]
and let
\[
m(s)=k_\mu'(\theta(m_0)-s).
\]
Then $\Phi'(s)=m(s)>0$, so $L(s)\leq1$. If $P$ placed positive mass on $(-\infty,0)$, then $L(s)$ would diverge exponentially as $s\to\infty$, a contradiction. Thus $P$ is carried by $[0,\infty)$.

Because $P$ is infinitely divisible and nonnegative, its Laplace transform has the L\'evy--Khintchine representation
\begin{equation}
L(s)=\exp\left\{-ds-\int_{(0,\infty)}(1-e^{-sx})\nu(dx)\right\},
\label{eq:LK}
\end{equation}
where $d\geq0$. Differentiation gives
\begin{equation}
m(s)=d+\int_{(0,\infty)}xe^{-sx}\nu(dx),
\label{eq:msLevy}
\end{equation}
while NEF calculus gives
\begin{equation}
m'(s)=-V(m(s)).
\label{eq:ODE}
\end{equation}
Since $m(s)>0$ is decreasing, let $\ell=\lim_{s\to\infty}m(s)\geq0$. If $\ell>0$, continuity and positivity of $V$ imply that, for all sufficiently large $s$, $V(m(s))\geq V(\ell)/2$. By \eqref{eq:ODE}, $m'(s)\leq-V(\ell)/2$, which would eventually force $m(s)<0$, a contradiction. Hence $m(s)\downarrow0$, and \eqref{eq:msLevy} gives $d=0$. Also,
\[
\Phi(s)=\int_{m(s)}^{m_0}\frac{m}{V(m)}\,dm.
\]
By \eqref{eq:pform}, the integrand is asymptotic to $1/(a_rm^{p-1})$. Since $p\geq2$, $\Phi(s)\to\infty$, and monotone convergence in \eqref{eq:LK} yields
\begin{equation}
\nu(0,\infty)=\infty.
\label{eq:infiniteactivity}
\end{equation}
Moreover,
\[
\int_{(0,\infty)}x\nu(dx)=m_0,
\qquad
\int_{(0,\infty)}x^2\nu(dx)=V(m_0).
\]
For $\operatorname{Re}z\geq0$ define
\[
M(z)=\int_{(0,\infty)}xe^{-zx}\nu(dx).
\]
Choose $m_0$ small enough that the closed disk $|w|\leq m_0$ lies in the domain of analyticity of $H$. Then $M$ is holomorphic for $\operatorname{Re}z>0$, continuous on its closure, and
\[
|M(z)|\leq\int_{(0,\infty)}x\nu(dx)=m_0.
\]
Thus $H(M(z))$ is holomorphic in the open right half-plane. For real $s\geq0$, $M(s)=m(s)$, so $M'(s)=-M(s)^pH(M(s))$. The identity theorem therefore gives
\begin{equation}
M'(z)=-M(z)^pH(M(z)),
\qquad \operatorname{Re}z>0.
\label{eq:complexODE}
\end{equation}
If $M$ had a zero of multiplicity $q\geq1$, the two sides of \eqref{eq:complexODE} would have zero multiplicities $q-1$ and $pq$, respectively, which is impossible because $p\geq2$. Hence $M$ is zero-free in the open right half-plane.

Taking $m_0$ smaller if necessary, assume $|H(w)-H(0)|\leq H(0)/2$ for $|w|\leq m_0$. Since $H(0)>0$, this gives $\operatorname{Re}H(w)\geq H(0)/2$ on that disk. Because $M$ is zero-free,
\[
Q(z)=M(z)^{-(p-1)}
\]
is holomorphic in the open right half-plane and satisfies $Q'(z)=(p-1)H(M(z))$. Passing to the boundary and integrating along the imaginary axis,
\[
Q(it)-Q(0)=i(p-1)\int_0^t H(M(iu))\,du.
\]
Hence $|\operatorname{Im}Q(it)|$ grows at least linearly in $|t|$, up to an additive constant, and therefore $|Q(it)|\geq c(1+|t|)$ for some $c>0$. Consequently,
\[
|M(it)|\leq C(1+|t|)^{-1/(p-1)}.
\]

Let $\lambda(dx)=x^2\nu(dx)$, which is finite. Boundary limits in $-M'(\varepsilon+it)$ give
\[
\widehat\lambda(t)=M(it)^pH(M(it)).
\]
Hence
\[
|\widehat\lambda(t)|\leq C(1+|t|)^{-p/(p-1)}.
\]
Since $p/(p-1)>1$, $\widehat\lambda\in L^1(\R)$. Fourier inversion gives a bounded continuous density $h$ for $\lambda$; see Rudin \cite{Rudin1962}. Consequently
\[
\nu(dx)=\frac{h(x)}{x^2}\,dx,
\qquad x>0.
\]
Thus $\nu\ll dx$. To pass from the L\'evy measure to the law, choose increasing Borel sets $A_n\subset(0,\infty)$ with $0<\nu(A_n)<\infty$ and $\nu(A_n)\uparrow\infty$. Split the infinitely divisible law as the convolution of the compound-Poisson law with L\'evy measure $\nu|_{A_n}$ and the law generated by the remaining L\'evy measure. Since $\nu|_{A_n}\ll dx$, the compound-Poisson factor is the sum of an atom at zero of mass $e^{-\nu(A_n)}$ and an absolutely continuous measure. Convolution shows that the singular mass of $P$ is at most $e^{-\nu(A_n)}$, which tends to zero by \eqref{eq:infiniteactivity}. Hence $P\ll dx$. Exponential tilting transfers absolute continuity to every member, and $P(\{0\})=0$. Since $m(s)\downarrow0$, the lower endpoint of the convex support is zero; a nondegenerate infinitely divisible law cannot have bounded support, so its upper endpoint is infinite. Hence
\[
C_\mu=[0,\infty).
\]
\end{proof}

\begin{corollary}
If
\[
G(m)=\sum_{j=0}^{r}a_jm^j,
\qquad a_j\geq0,
\]
is a nonzero finite polynomial, then $V(m)=m^2G(m)$ is the variance function of an infinitely divisible NEF with mean domain $\R_+$, every member is absolutely continuous on $[0,\infty)$, and $C_\mu=[0,\infty)$. Hence the family is steep. Finite polynomials of arbitrary degree therefore produce continuous NEFs on the full positive mean domain.
\end{corollary}

\begin{corollary}
Let
\[
V(m)=\sum_{j=1}^{r}a_jm^j,
\qquad a_j\geq0,
\]
be nonzero. If $a_1>0$, set $Y=X/a_1$. Then $V_Y$ has linear coefficient one, so the result of Bar-Lev, Letac and Ridder \cite{BarLevLetacRidder2024} implies that $Y$ is supported on $\N_0$; hence $X$ is supported on $a_1\N_0$. If $a_1=0$, then $V(m)=m^2G(m)$ with $G$ a nonnegative polynomial, and every member is absolutely continuous on $[0,\infty)$. Thus the polynomial class has an exhaustive counting-versus-continuous dichotomy according as $a_1>0$ or $a_1=0$.
\end{corollary}

\section{Examples and consequences}

\noindent\textit{Classical examples.}
For $G(m)=a$, $a>0$, one has $V(m)=am^2$, the gamma NEF. The member with mean $m$ has density
\[
f(x;m)=\frac{x^{1/a-1}e^{-x/(am)}}{\Gamma(1/a)(am)^{1/a}},\qquad x>0.
\]
For $G(m)=am$, $V(m)=am^3$, giving the inverse Gaussian NEF. For $G(m)=1+m$, one obtains the Kendall--Ressel NEF,
\[
V(m)=m^2(1+m),
\]
whose generating density can be written
\[
h(x)=\frac{x^xe^{-x}}{\Gamma(x+2)},\qquad x>0;
\]
see Letac and Mora \cite{LetacMora1990}. More generally, every integer power model $V(m)=am^p$, $p=2,3,\ldots$, belongs to the class through $G(m)=am^{p-2}$; for $p>2$ these are positive-stable-generated power-variance families \cite{BarLevEnis1986}.

For every $c>0$,
\[
V_c(m)=m^2(1+cm^2)=m^2+cm^4
\]
is an infinitely divisible, absolutely continuous NEF with mean domain $\R_+$. More generally,
\[
V(m)=m^2(1+c_1m+\cdots+c_rm^r),\qquad c_j\geq0,
\]
provides such a family for every degree $r$, and infinite power series give an infinite-dimensional extension.

All cumulants follow directly from $V$.

\begin{proposition}
Let $\kappa_j(m)$ denote the $j$th cumulant of the member with mean $m$. Then
\[
\kappa_1(m)=m,
\qquad
\kappa_2(m)=V(m),
\]
and for every $j\geq2$,
\begin{equation}
\kappa_{j+1}(m)=V(m)\frac{d}{dm}\kappa_j(m).
\label{eq:cumulantrec}
\end{equation}
In particular,
\[
\kappa_3(m)=V(m)V'(m),
\qquad
\kappa_4(m)=V(m)[V'(m)]^2+V(m)^2V''(m).
\]
\end{proposition}

This follows from $d/d\theta=V(m)d/dm$.

For the quartic example $V(m)=m^2(1+cm^2)$, one obtains, for example,
\[
\kappa_3(m)=2m^3(1+cm^2)(1+2cm^2),
\]
and higher cumulants follow recursively.

\begin{remark}[Finite radius and steepness]
If $G$ is an infinite power series with finite radius $R$, the construction gives $M_F=(0,R)$ while Theorem 1 gives $C_\mu=[0,\infty)$. Hence
\[
M_F\subsetneq\operatorname{int}C_\mu=(0,\infty),
\]
and the NEF is non-steep. By contrast, finite-polynomial $G$ has $R=\infty$, so $M_F=(0,\infty)$ and the family is steep.
\end{remark}

\subsection*{Statistical and Bayesian remarks}
These remarks are consequences, not additional novelty claims. The recursion \eqref{eq:cumulantrec} gives all cumulants directly from $V$. Moreover,
\[
\frac{d\theta}{dm}=\frac{1}{V(m)},\qquad \frac{dk}{dm}=\frac{m}{V(m)},
\]
so likelihood calculations and the Diaconis--Ylvisaker conjugate analysis \cite{DiaconisYlvisaker1979} can be carried out from the variance function even when the generating density is not available in closed form. Further inferential development is beyond the scope of this Letter.

\section{Discussion}

Bar-Lev \cite{BarLev1987} supplied a general existence framework based on absolute monotonicity. Within it, the subclasses $mG(m)$ and $m^2G(m)$ provide a useful discrete-versus-continuous comparison, without being claimed to exhaust the full 1987 construction. Bar-Lev, Letac and Ridder \cite{BarLevLetacRidder2024} identified the normalized first-order subclass as counting; the present result identifies the analytic nonnegative-coefficient subclass $m^2G(m)$ as absolutely continuous.

For the polynomial class the conclusion is exhaustive: a positive linear coefficient can be scaled to one and gives support on $a_1\N_0$, whereas a zero linear coefficient forces a factor $m^2$ and hence the continuous case. Finite-polynomial $G$ gives $M_F=(0,\infty)=\operatorname{int}C_\mu$ and hence a steep family; finite-radius nonpolynomial $G$ gives a non-steep family with the same positive convex support. The variance function gives all cumulants recursively and determines the canonical and cumulant functions needed for likelihood and conjugate Bayesian calculations; broader inferential questions are left for future work.

\end{document}